\documentclass[11pt]{llncs} 
\usepackage{amsmath,latexsym}
\usepackage{amsfonts} 
\usepackage{amssymb}
\usepackage{graphicx}
\usepackage{subfigure}
\usepackage{color} 
\usepackage{epstopdf}
\usepackage{setspace}
\usepackage{algpseudocode}
\usepackage{algorithm}
\usepackage{cite} 
\usepackage{datetime}
\def\qedbox#1#2{\vbox{\hrule height.2pt
  \hbox{\vrule width.2pt height#2pt \kern#1pt \vrule width.2pt}
  \hrule height.2pt}}
\def\qed{\hfill \quad\qedbox46\newline\smallbreak}

\def\s#1{\mbox{\boldmath $#1$}}

\def\+{\!+\!}
\def\-{\!-\!}

\def\itbf#1{\textit{\textbf{#1}}}

\def\bcontinue{{\bf continue}}

\def\pref(#1,#2){$#1$ is a prefix of $#2$}
\def\suff(#1,#2){$#1$ is a suffix of $#2$}

\def\reg(#1,#2){$#2$ is $#1$-regular}
\def\notreg(#1,#2){$#2$ is not $#1$-regular}

\def\UPDATE\_F{\tt{UPDATE\_F}}

\def\EL{\mathcal L}
\def\EC{\mathcal C}

\algnewcommand{\LineComment}[1]{\State \(\triangleright\) \normalfont{\sl #1}}
\algtext*{EndWhile}
\algtext*{EndIf}
\algtext*{EndFor}
\algtext*{EndFor}
\algtext*{EndProcedure}

\begin{document} 
\pagestyle{headings}
 \title{\normalsize 
 String Comparison in $V$-Order: \\ New Lexicographic Properties \& On-line
 Applications}

\author{
Ali Alatabbi\inst{1}
\and
Jacqueline W.\ Daykin\inst{1,2}
\and 
M.\ Sohel Rahman\inst{3}   
\and\\
W.\ F.\ Smyth\inst{1,4,5}\thanks{This work was supported in part by the
Natural Sciences \& Engineering Research Council of Canada.}}

\institute{Department of Informatics, King's College London, UK\\
\email{ali.alatabbi@kcl.ac.uk, jackie.daykin@kcl.ac.uk}
\and Department of Computer Science\\
Royal Holloway College, University of London, UK\\
\email{J.Daykin@cs.rhul.ac.uk}
\and 
A$\ell$EDA Group, Department of CSE, BUET, Dhaka-1000, Bangladesh\\
\email{msrahman@cse.buet.ac.bd}
\and Algorithms Research Group, Department of Computing \& Software\\
McMaster University, Canada\\
\email{smyth@mcmaster.ca}
\and School of Engineering \& Information Technology\\
Murdoch University, Western Australia}

\maketitle

\begin{abstract} 
$V$-order is a global order on strings related to Unique Maximal
Factorization Families (UMFFs) \cite{DDS11,DDS13}, which are themselves
generalizations of Lyndon words \cite{L83}. $V$-order has recently been
proposed as an alternative to lexicographical order in the computation of
suffix arrays and in the suffix-sorting induced by the Burrows-Wheeler
transform. Efficient $V$-ordering of strings thus becomes a matter of
considerable interest. In this paper we present new and surprising results 
on $V$-order in strings, then go on to explore the algorithmic consequences.  
\end{abstract}

\section{Introduction}\label{sect-intro}

This paper extends current knowledge on the non-lexicographic string
ordering technique known as $V$-order \cite{DaD96}. New combinatorial
insights are obtained which are linked to computational settings. In
particular, we relate $V$-order string comparison to lexicographic by
showing how it is possible to traverse the strings from left to right,
respectively right to left, at each stage determining in $O(1)$ time the
order of prefixes, respectively suffixes.  This improves on existing
ordering algorithms \cite{DDS11, ADRS14, ADRS14FI} in various ways: it
removes any dependence on an ``indexed" alphabet, it orders prefixes and
suffixes in addition to the original strings, and it reduces dependence on
additional data structures. Furthermore, we introduce an input-sensitive
variant for $V$-order comparison.

Regarding practical applications of $V$-order, in \cite{DS14} a novel
variant of the classic lexicographic Burrows-Wheeler transform, the
$V$-transform ($V$-BWT), was introduced which was based on $V$-order --
instances of enhanced data clustering were demonstrated. Linear $V$-sorting
of all the rotations of a string $\s{x} = \s{x}[1 \ldots n]$, as required
for an efficient transform, was achieved by linear time and space $V$-order
string comparison (Daykin et al. 2011) \cite{DDS11} along with $\Theta (n)$
suffix-sorting (Ko and Aluru, 2003) \cite{KA03}. Lyndon-like factorization
of a string into $V$-words is likewise  linear in time and space
\cite{DDS11}. For $V$-words, \cite{DS14} showed how to compute the
$V$-transform in $\Theta (n)$ time and space; in addition, inverting the
$V$-transform to recover the input $V$-word was achieved in time $O(n^2
\log k')$, using $O(n+k')$ additional storage, where $k'$ is the number of
sequences of largest letters in \s{x}. A bijective algorithm was also
outlined in the case that \s{x} is arbitrary.

We apply the new combinatorial insights gained to modify ideas given in
\cite{MRRS14} for Lyndon factorizations, suffix arrays and the Burrows
Wheeler transform, to similarly obtain on-line processing for $V$-order.

\section{Preliminaries}\label{sect-prelim}

Consider a finite totally ordered alphabet $\Sigma$ which consists of a set
of characters (equivalently letters or symbols) with cardinality
$|\Sigma|$.
A string is a sequence of zero or more characters over $\Sigma$.
A string $\s{s}$ of length $|\s{s}|=n$ is represented by $\s{s}[1\ldots
n]$, where $\s{s}[i] \in \Sigma $ for $ 1\leq i \leq n $.
The set of all non-empty strings over the alphabet $\Sigma$ is denoted by
$\Sigma^+$. The empty string with zero length is denoted by $\epsilon$, with
$\Sigma^* = \Sigma^+ \cup \epsilon$; A string $\s{w}$ is a substring, or
factor, of $\s{s}$ if $\s{s} = \s{u}\s{w}\s{v}$, where $\s{u},\s{v} \in
\Sigma^{\ast}$. Words $\s{w}[1\ldots i]$ are prefixes of $\s{w}$, and words
$\s{w}[i \ldots n]$ are suffixes of $\s{w}$. For further stringological
definitions, theory and algorithmics see \cite{CR02}.

Some of our applications are derived from Lyndon words, which we now
introduce. A string $\s{y}=\s{y}[1\ldots n]$ is a \emph{conjugate} (or
cyclic rotation) of $\s{x}=\s{x}[1 \ldots n]$ if $\s{y}[1 \ldots n] =
\s{x}[i \ldots n]\s{x}[1\ldots i-1]$ for some $1 \leq i \leq n$ (for $i
=1, ~\s{y} =\s{x}$).
A {\itshape Lyndon word} is a primitive word which is minimal for the
lexicographical order (lexorder) of its conjugacy class. 

\begin{theorem}\label{lLyndon-thrm}
\cite{CFL58} Any word $\s{w}$ can be written uniquely as a non-increasing
product $\s{w} = \s{u}_1 \s{u}_2 \cdots \s{u}_k$ of Lyndon words.
\end{theorem}

Theorem \ref{lLyndon-thrm} shows that there is a unique decomposition of
any word into non-increasing Lyndon words $(\s{u}_1 \ge \s{u}_2 \ge \cdots
\ge \s{u}_k)$. We proceed to define a non-lexicographic order, $V$-order,
and then establish useful new lexicographic characteristics for $V$-order. 

Let $\s{x}=x_1x_2\cdots x_n$ be a
string over $\Sigma$. Define $h \in \{1,\ldots,n\}$ by $h = 1$ if $x_1 \le x_2 \le \cdots
 \le x_n$; otherwise, by the unique value such that $x_{h-1}>x_h \le x_{h+1} \le x_{h+2}
\le \cdots \le x_n$. Let $\s{x}^*=x_1x_2\cdots x_{h-1}x_{h+1}\cdots x_n$, where
the star * indicates deletion of the letter $x_h$. Write $\s{x}^{s*}$ for
$(...(\s{x}^*)^{*}...)^{*}$ with $s \ge 0$ stars. Let $g =
\max\{x_1,x_2,...,x_n\}$, and let $k$ be the number of occurrences of
$g$ in $\s{x}$. Then the sequence $\s{x},\s{x}^*,\s{x}^{2*},...$
ends with $g^{k},...,g^2,g^1,g^0=\s{\varepsilon}$. In the {\it star tree}
each string $\s{x}$ over $\Sigma$ labels a vertex, and there is a directed edge from $\s{x}$
to $\s{x}^*$, with the empty string $\s{\varepsilon}$ as the root.

\begin{definition}
\label{def-V_order}
We define {\rm $V$-order} $\prec$ between distinct strings \s{x},\s{y} with $\s{x} \prec \s{y}$.
First $\s{x} \prec \s{y}$ if \s{x} is in the path
$\s{y},\s{y}^*,\s{y}^{2*},...,\s{\varepsilon}$. If $\s{x},\s{y}$ are not in a path, there exist
smallest $s,t$ such that $\s{x}^{(s+1)*}=\s{y}^{(t+1)*}$. Put $\s{s}=\s{x}^{s*}$ and
$\s{t}=\s{y}^{t*}$; then $\s{s} \neq \s{t}$ but $|\s{s}| = |\s{t}| = m$ say.
Let $j \in 1..m$ be the greatest integer such that $\s{s}[j] \ne \s{t}[j]$.
If $\s{s}[j]<\s{t}[j]$ in $\Sigma$ then $\s{x} \prec \s{y}$. Clearly $\prec$ is a total
order.
\end{definition}

For instance, using the natural ordering of integers,  if $\s{x} = 32415$, then
$\s{x}^{*} = 3245$, $\s{x}^{2*} = 345$, $\s{x}^{3*} = 45$ and so $45 \prec 32415$.

\begin{definition}\label{Vform}{\cite{DaD96, DD03, DDS11, DDS13}}
The \itbf{$V$-form} of a string \s{x} is defined as
$$V_k(\s{x}) = \s{x} = \s{x_0}g\s{x_1}g\cdots\s{x_{k-1}}g\s{x_k}$$

\noindent for strings $\s{x_i},\ i = 0,1,\ldots,k$, where $g$ is the largest
letter in \s{x} --- thus we suppose that $g$ occurs exactly $k$ times.
For clarity, when more than one string is involved, we use the notation $g =
\mathcal{L}_{\s{x}}$, $k = \mathcal{C}_{\s{x}}$.
\end{definition}

\begin{lemma}\label{lem-a}
{\cite{DaD96, DD03, DDS11,DDS13}} Suppose we are given distinct strings
$\s{x}$ and $\s{y}$ with corresponding $V$-forms as follows:
\begin{eqnarray}
\s{x} & = & \s{x}_0 \mathcal \EL_{\s{x}} \s{x}_1 \mathcal \EL_{\s{x}} \s{x}_2 \cdots
\s{x}_{j-1} \mathcal \EL_{\s{x}} \s{x}_{j}, \nonumber \\
\s{y} & = & \s{y}_0 \mathcal \EL_{\s{y}} \s{y}_1 \mathcal \EL_{\s{y}} \s{y}_2 \cdots
\s{y}_{k-1} \mathcal \EL_{\s{y}} \s{y}_{k}, \nonumber
\end{eqnarray}
where $j = \mathcal{C}_{\s{x}},\ k = \mathcal{C}_{\s{y}}$.

Let $h \in \{0 \ldots \max(j,k)\}$ be the least integer such that $\s{x}_h
\neq \s{y}_h$. Then  $\s{x} \prec \s{y}$ if, and only if, one of the
following conditions holds:
\begin{description}
\item [(C1)] $\mathcal \EL_{\s{x}} < \mathcal \EL_{\s{y}}$\label{lem-a-c1}
\item [(C2)] $\mathcal \EL_{\s{x}} = \mathcal \EL_{\s{y}}$ and
$\mathcal{C}_{\s{x}} < \mathcal{C}_{\s{y}}$ \label{lem-a-c2}
\item [(C3)] $\mathcal \EL_{\s{x}} = \mathcal \EL_{\s{y}}$,
$\mathcal{C}_{\s{x}} = \mathcal{C}_{\s{y}}$
and $\s{x}_h \prec \s{y}_h$.\label{lem-a-c3}
\end{description}
\end{lemma}

\begin{lemma}\label{lem-subsequce}
{\cite{DDS11,DDS13}} For given strings $\s{v}$ and \s{x},
if \s{v} is a proper subsequence of \s{x},
then $\s{v}\prec \s{x}$. 
\end{lemma}

\begin{example}\label{ex-lex+V_order}

We compare two dictionaries for a set of English words over the ordered Roman alphabet.

\noindent Lexorder($<$) dictionary: 
$catastrophe ~<~ sop ~<~ strop ~<~ strophe ~<~ top.$

\noindent The well-known lexorder positional technique seeks the first
difference from the left and then applies the ordering of the alphabet.

\noindent $V$-order ($\prec$) dictionary: 
$sop ~\prec~ top ~\prec~ strop ~\prec~ strophe ~\prec~ catastrophe.$

\noindent The first $V$-order comparison is determined by Lemma \ref{lem-a}(C1)
and the following three by the useful Lemma \ref{lem-subsequce}.
\end{example}

\section{New Results on $V$-Order}\label{sect-results}
A main interest of this paper is to consider positional lexorder-type ordering
techniques for $V$-order, for which we first establish some basics.
Given an ordered alphabet $\Sigma=\{1<2< \cdots \}$ and a string $\s{x} \in
\Sigma^+$ with $|\s{x}| > 1$, then from conditions (C1, C2) we have, as for lexorder,
$1 \prec \s{x} \prec \s{x^i}$ for all $i > 1$. For strings $\s{u}, \s{v}, \s{w}
\in \Sigma^+$ with $\s{u} \prec \s{v} \prec \s{w}$, we find by Lemma
\ref{lem-subsequce} that, again as for lexorder, both $\s{u} \prec \s{u}\s{v}$
and $\s{vw} \nprec \s{w}$ (in contrast to Lyndon words).
In general, for $i,j >1$, we can say that
$$ 1 \prec \s{u} \prec \s{u^2} \prec \cdots \prec \s{u^i} \prec \s{u^i v} \prec
\cdots \prec \s{u^i v^j} \prec \cdots \prec \s{u^i v^jw} \prec \cdots $$
We begin by generalizing Lemma 2.5 in \cite{DS14}:

\begin{lemma}\label{lem-unexpected1}
For any two strings \s{x}, \s{y} and $\lambda \in \Sigma$,
$\s{x} \prec \s{y}  \Leftrightarrow  \s{x} \lambda \prec \s{y}\lambda.$
\end{lemma}

\begin{proof}
Let $\s{x'} = \s{x}\lambda$, $\s{y'} = \s{y}\lambda$.
First observe that if $\mathcal{L}_{\s{x}} < \mathcal{L}_{\s{y}}$,
then by (C1), $\s{x} \prec \s{y}$.
Furthermore:
\begin{itemize}
\item[$\bullet$]
if $\lambda < \mathcal{L}_{\s{y}}$, then $\s{x'} \prec \s{y'}$ by (C1), because
$ \mathcal{L}_{\s{x}} \le \mathcal{L}_{\s{x'}} < \mathcal{L}_{\s{y}} = \mathcal{L}_{\s{y'}}$;
\item[$\bullet$]
if $\lambda = \mathcal{L}_{\s{y}}$, then $\s{x'} \prec \s{y'}$ by (C2),
because $\mathcal{L}_{\s{x'}} = \mathcal{L}_{\s{y'}} = \lambda$ and
$\mathcal{C}_{\s{x'}} = 1 < \mathcal{C}_{\s{y'}}$;
\item[$\bullet$]
if $\lambda > \mathcal{L}_{\s{y}}$, then $\s{x'} \prec \s{y'}$ by (C3),
because $\mathcal{L}_{\s{x'}} = \mathcal{L}_{\s{y'}} = \lambda$,
$\mathcal{C}_{\s{x'}} = \mathcal{C}_{\s{y'}} = 1$, and $\s{x} \prec \s{y}$.
\end{itemize}
Thus the lemma holds for $\mathcal{L}_{\s{x}} < \mathcal{L}_{\s{y}}$
and, by the complementary argument, it holds also for
$\mathcal{L}_{\s{y}} < \mathcal{L}_{\s{x}}$.
We may assume therefore that $\mathcal{L}_{\s{x}} = \mathcal{L}_{\s{y}}$.

Suppose then that $\mathcal{C}_{\s{x}} < \mathcal{C}_{\s{y}}$,
so that by (C2), $\s{x} \prec \s{y}$.
Furthermore:
\begin{itemize}
\item[$\bullet$]
if $\lambda \le \mathcal{L}_{\s{x}} = \mathcal{L}_{\s{y}}$,
then $\s{x'} \prec \s{y'}$ by (C2),
because $\mathcal{C}_{\s{x'}} = \mathcal{C}_{\s{x}}\+ \delta <
\mathcal{C}_{\s{y}}\+ \delta = \mathcal{C}_{\s{y'}}$, where $\delta = 0$
($\lambda < \mathcal{L}_{\s{x}}$) or 1 ($\lambda = \mathcal{L}_{\s{x}}$);
\item[$\bullet$]
if $\lambda > \mathcal{L}_{\s{x}}$,
then $\s{x'} \prec \s{y'}$ by (C3),
because $\mathcal{L}_{\s{x'}} = \mathcal{L}_{\s{y'}} = \lambda$,
$\mathcal{C}_{\s{x'}} = \mathcal{C}_{\s{y'}} = 1$, and $\s{x} \prec \s{y}$.
\end{itemize}
Thus the lemma holds for $\mathcal{C}_{\s{x}} < \mathcal{C}_{\s{y}}$,
and as above also for $\mathcal{C}_{\s{y}} < \mathcal{C}_{\s{x}}$.

Suppose therefore that $\mathcal{L}_{\s{x}} = \mathcal{L}_{\s{y}}$,
$\mathcal{C}_{\s{x}} = \mathcal{C}_{\s{y}}$.
Then whether or not
$\s{x} \prec \s{y}$ depends on the least value $h$ of Lemma~\ref{lem-a}
such that $\s{x}_h \prec \s{y}_h$ or $\s{y}_h \prec \s{x}_h$:
\begin{itemize}
\item[$\bullet$]
If $\lambda = \mathcal{L}_{\s{x}} = \mathcal{L}_{\s{y}}$,
then $h$ is unchanged by appending $\lambda$
to \s{x} and to \s{y}, so that, in this case,
$\s{x} \prec \s{y} \Leftrightarrow \s{x'} \prec \s{y'}$, as required.
\item[$\bullet$]
For $\lambda > \mathcal{L}_{\s{x}}$,
we find as above that $\mathcal{L}_{\s{x'}} = \mathcal{L}_{\s{y'}} = \lambda$,
$\mathcal{C}_{\s{x'}} = \mathcal{C}_{\s{y'}} = 1$, the ordering of \s{x'} and
\s{y'} is equivalent to the ordering of \s{x} and \s{y}.
\item[$\bullet$]
Finally, suppose that $\lambda < \mathcal{L}_{\s{x}} = \mathcal{L}_{\s{y}}$.
If $h < \mathcal{C}_{\s{x}}$, then as above the ordering of \s{x'},\s{y'}
corresponds to the ordering of \s{x},\s{y},
unaffected by appending $\lambda$.
If on the other hand $h = \mathcal{C}_{\s{x}}$,
then the problem reduces recursively to ordering $\s{x_h}\lambda,\s{y_h}\lambda$
based on the ordering of \s{x_h},\s{y_h},
where $\mathcal{L}_{\s{x_h}} < \mathcal{L}_{\s{x}}$ 
and $\mathcal{L}_{\s{y_h}} < \mathcal{L}_{\s{y}}$.
Thus, after a finite number of such reductions,
one of the above cases must hold.
\end{itemize}

This completes the proof.  \qed
\end{proof}

\begin{lemma}\label{lem-unexpected2}
For any two strings \s{x}, \s{y} and $\lambda \in \Sigma$,
$\s{x} \prec \s{y} \Leftrightarrow \lambda\s{x} \prec \lambda\s{y}.$
\end{lemma}
\begin{proof}
The argument is analogous to that given for Lemma~\ref{lem-unexpected1}. Note that the recursive case $\lambda \s{x_0}, \lambda \s{y_0}$ is likewise based on the ordering of $\s{x_0}, \s{y_0}$,  where $\mathcal{L}_{\s{x_0}} < \mathcal{L}_{\s{x}}$
and $\mathcal{L}_{\s{y_0}} < \mathcal{L}_{\s{y}}$. \qed
\end{proof}

Interestingly, although Lemma \ref{lem-unexpected2} holds for lexorder, Lemma
\ref{lem-unexpected1} does not as shown by: $a < ab$ in lexorder but $ac \nless
abc$.

We can now combine the above lemmas into a more general result:

\begin{theorem}
\label{thrm-presuf}
For any strings \s{u}, \s{v}, \s{x}, \s{y},
$\s{x} \prec \s{y} \Leftrightarrow \s{uxv} \prec \s{uyv}$.
\end{theorem}

\begin{proof}
This follows from repeated applications of Lemmas \ref{lem-unexpected1} \& \ref{lem-unexpected2},
where we append one letter at a time to suffixes
and prepend one letter at a time to prefixes.
\qed
\end{proof}
 
We can establish extensions and applications of these results:

\begin{lemma}\label{lem-unexpected5}\label{lem-addLambda}
Let \s{x} and \s{y} be strings with $V$-forms
\begin{eqnarray}
\s{x} & = & \s{x}_0  \EL_{\s{x}} \s{x}_1  \EL_{\s{x}} \s{x}_2 \cdots
\s{x}_{j-1} \mathcal \EL_{\s{x}} \s{x}_{j}, \nonumber \\
\s{y} & = & \s{y}_0 \EL_{\s{y}} \s{y}_1  \EL_{\s{y}} \s{y}_2 \cdots
\s{y}_{k-1} \EL_{\s{y}} \s{y}_{k}. \nonumber
\end{eqnarray}
For any letter $\lambda \le \max(\EL_{\s{x}},\EL_{\s{y}})$
and any integer $i \in \{0 \ldots \max(j,k)\}$, let
\begin{eqnarray}
\s{x'}  = &  \s{x}_0  \EL_{\s{x}} \cdots \EL_{\s{x}}
\s{x}_i \lambda \EL_{\s{x}} \cdots  \EL_{\s{x}} \s{x}_{j}
\nonumber,\\
\s{y'}  = &  \s{y}_0  \EL_{\s{y}} \cdots \EL_{\s{y}}
\s{y}_i \lambda \EL_{\s{y}} \cdots  \EL_{\s{y}} \s{y}_{k}
\nonumber,\\
\s{x''}  = & \s{x}_0  \EL_{\s{x}} \cdots \EL_{\s{x}} \lambda
\s{x}_i \EL_{\s{x}} \cdots  \EL_{\s{x}} \s{x}_{j} \nonumber,\\
\s{y''} = & \s{y}_0  \EL_{\s{y}}
\cdots \EL_{\s{y}} \lambda\s{y}_i \EL_{\s{y}} \cdots  \EL_{\s{y}}
\s{y}_{k} \nonumber.
\end{eqnarray}
Then $\s{x'} \prec \s{y'} \Leftrightarrow \s{x} \prec \s{y} \Leftrightarrow \s{x''} \prec \s{y''}$.
\end{lemma}

\begin{proof}
First suppose that $\s{x'} \prec \s{y'}$, so that one of the conditions
(C1)-(C3) of Lemma~\ref{lem-a} must hold:
\begin{itemize}
\item[$\bullet$]
Assume that $\EL_{\s{x'}} < \EL_{\s{y'}}$.
Then $\lambda < \EL_{\s{y}}$ and
$\EL_{\s{x}} \le \EL_{\s{x'}} < \EL_{\s{y'}} = \EL_{\s{y}}$,
so that $\s{x} \prec \s{y}$ by (C1).
\item[$\bullet$]
Assume that $\EL_{\s{x'}} = \EL_{\s{y'}}$, with $\EC_{\s{x'}} < \EC_{\s{y'}}$.
If $\lambda = \EL_{\s{y}}$,
then either $\EL_{\s{x}} < \EL_{\s{y}}$ or
$\lambda = \EL_{\s{x}}$ and
$\EC_{\s{x}} = \EC_{\s{x'}} - 1 < \EC_{\s{y'}} - 1 = \EC_{\s{y}}$;
otherwise,
$\lambda < \EL_{\s{y}}$, so that $\EL_{\s{x}} = \EL_{\s{y}}$
with $\EC_{\s{x}} = \EC_{\s{x'}} < \EC_{\s{y'}} = \EC_{\s{y}}$.
In all three cases, $\s{x} \prec \s{y}$ by (C2).
\item[$\bullet$]
If $\EL_{\s{x'}} = \EL_{\s{y'}}$ and $\mathcal{C}_{\s{x'}} =
\mathcal{C}_{\s{y'}}$, then whether or not
$\s{x} \prec \s{y}$ depends on the least value $h$ of Lemma 1
such that $\s{x}_h \prec \s{y}_h$:
\begin{itemize}
\item [$\circ$]
if $h\not=i$, then the ordering of $\s{x},\s{y}$
corresponds to the ordering of $\s{x'},\s{y'}$,
unaffected by removing $\lambda$;
\item [$\circ$]
if $h = i$, then the ordering of $\s{x},\s{y}$
reduces to the ordering of $\s{x_h}\lambda,\s{y_h}\lambda$, so that
$\s{x} \prec \s{y}$ by Theorem 1.
\end{itemize}
\end{itemize}
Next suppose that $\s{x} \prec \s{y}$.
Again we consider the conditions (C1)-(C3) of Lemma~\ref{lem-a}:
\begin{itemize}
\item[$\bullet$]
Assume that $\EL_{\s{x}} < \EL_{\s{y}}$.
If $\lambda = \EL_{\s{y}}$, then $\lambda = \EL_{\s{x'}} = \EL_{\s{y'}}$
with $\EC_{\s{x'}} = 1 < \EC_{\s{y'}}$,
so that $\s{x'} \prec \s{y'}$ by (C2);
while if
$\lambda < \EL_{\s{y}}$, then $\s{x'} \prec \s{y'}$ by (C1),
because $ \EL_{\s{x}} \le \EL_{\s{x'}} < \EL_{\s{y}} = \EL_{\s{y'}}$.
\item[$\bullet$]
Assume that $\EL_{\s{x}} = \EL_{\s{y}}$, with $\EC_{\s{x}} < \EC_{\s{y}}$.
If $\lambda = \EL_{\s{x}} = \EL_{\s{y}}$, then $\mathcal{C}_{\s{x'}} =
 \mathcal{C}_{\s{x}}\+ 1 < \mathcal{C}_{\s{y}}\+ 1 = \mathcal{C}_{\s{y'}}$;
if $\lambda < \EL_{\s{x}} = \EL_{\s{y}}$,
then $\mathcal{C}_{\s{x'}} = \mathcal{C}_{\s{x}} < \mathcal{C}_{\s{y}} =
\mathcal{C}_{\s{y'}}$.
In both cases,
 $\s{x'} \prec \s{y'}$ by (C2).
\item[$\bullet$]
If $\EL_{\s{x}} = \EL_{\s{y}}$ and $\mathcal{C}_{\s{x}} =
\mathcal{C}_{\s{y}}$, then again whether or not $\s{x'} \prec \s{y'}$ depends on
the least value $h$ of Lemma~\ref{lem-a} such that $\s{x}_h \prec \s{y}_h$:
\begin{itemize}
\item[$\circ$]
if $h\not=i$, then the ordering of $\s{x'},\s{y'}$
corresponds to the ordering of $\s{x},\s{y}$,
unaffected by adding $\lambda$;
\item[$\circ$]
if $h = i$, then the ordering of $\s{x'},\s{y'}$
reduces to the ordering of $\s{x_h}\lambda,\s{y_h}\lambda$, so that
$\s{x'} \prec \s{y'}$ by Theorem \ref{thrm-presuf}.
\end{itemize}
\end{itemize}
This completes the proof that $\s{x'} \prec \s{y'} \Leftrightarrow \s{x} \prec \s{y}$.
The proof that $\s{x''} \prec \s{y''} \Leftrightarrow \s{x} \prec \s{y}$ is similar.
\qed
\end{proof}

To see that Lemma~\ref{lem-unexpected5}
does not hold for $\lambda > \max(\EL_{\s{x}},\EL_{\s{y}})$, consider
$$\s{x} = 1323 \prec \s{y} = 3133,\ \lambda = 4,\ \mbox{but } \s{y'} = 43133 \prec \s{x'} = 14323.$$
\begin{remark}
\label{rem-u}
Lemma~\ref{lem-unexpected5} is easily generalized by replacing $\lambda$
by any string $\s{u} = u_1u_2\cdots u_m$ such that,
for $1 \le j \le m$, $u_m \le \max(\EL_{\s{x}},\EL_{\s{y}})$,
and inserting such a \s{u} at any or all positions
$i \in \{0 \ldots \max(j,k)\}$.
\end{remark}

\begin{lemma}\label{lem-unexpected2-extended}
For any two strings $\s{x}, \s{y}$ and letters $\lambda, \mu \in \Sigma$, $\lambda \le \mu$:
\begin{description}
\item [(i)] $\s{x} \prec \s{y} \Rightarrow \lambda \s{x} \prec \mu \s{y}$;
\item [(ii)] $\s{x} \prec \s{y} \Rightarrow \s{x} \lambda \prec \s{y} \mu$.
\end{description}
\end{lemma}

\begin{proof}
For $\lambda = \mu$, (i) reduces to Lemma~\ref{lem-unexpected2},
while (ii) reduces to Lemma~\ref{lem-unexpected1}.
Thus we may assume $\lambda < \mu$.

Suppose $\s{x} \prec \s{y}$.
Then by Lemma~\ref{lem-unexpected2}
$\lambda\s{x} \prec \lambda\s{y}$,
while by Theorem~\ref{thrm-presuf} with $\s{u} = \s{\varepsilon}$,
$\lambda\s{y} \prec \mu\s{y}$.
Therefore $\lambda\s{x} \prec \mu\s{y}$, proving (i).
The proof of (ii) is similar.  \qed
\end{proof}
 
The following examples show that sufficiency does not hold in
Lemma~\ref{lem-unexpected2-extended}:
\begin{description}
\item[(i)]
$\s{y} = 441 \prec \s{x} = 442$,
$\lambda = 2 < \mu = 3$, but
$\lambda\s{x} = 2442 \prec \mu\s{y} = 3441$;
\item[(ii)]
$\s{y} = 441 \prec \s{x} = 442$,
$\lambda = 2 < \mu = 3$, but
$\s{x}\lambda = 4422 \prec \s{y}\mu = 4413$.
\end{description}

\section{Applications}\label{sect-appls}

Some of the results presented above lead us to some interesting applications. In
this section, we first present a brief sketch of an idea for a new string
comparison algorithm in $V$-order and then proceed to consider applications of
our results to suffix arrays (SAs) and the Burrows Wheeler transform (BWT).

\subsection{$V$-Order String Comparison}\label{sec:newAlgo}
Recently, Alatabbi et al. presented an interesting $V$-order string comparison
algorithm in \cite{ADRS14,ADRS14FI} (referred to as the ADRS algorithm
henceforth), where a mapping of the position of each letter in the string is
exploited to check for the conditions stated in Lemma \ref{lem-a}. Note that
there are three conditions in Lemma \ref{lem-a} and things get most interesting
when we reach Condition (C3) because of its recursive nature. Now, the
efficiency of ADRS algorithm depends on a key result (cf. Corollary 2.9 of
\cite{ADRS14FI}) which proves that the mismatch position of the two strings
under comparison remains the same as we go deep into the recursion. This fact
along with the result presented in Lemma \ref{lem-addLambda} gives us yet
another idea for an efficient string comparison algorithm in $V$-order.
Essentially, the idea builds upon the idea of the map in the ADRS algorithm as
we will now outline.

Suppose we are given two strings, \s{x} and \s{y}, with $V$-forms
\begin{eqnarray}
\s{x} & = & \s{x}_0  \EL_{\s{x}} \s{x}_1  \EL_{\s{x}} \s{x}_2 \cdots
\s{x}_{j-1} \mathcal \EL_{\s{x}} \s{x}_{j}, \nonumber \\
\s{y} & = & \s{y}_0 \EL_{\s{y}} \s{y}_1  \EL_{\s{y}} \s{y}_2 \cdots
\s{y}_{k-1} \EL_{\s{y}} \s{y}_{k}. \nonumber
\end{eqnarray}

\begin{description}
\item[Step 1:]
We first scan the input strings from left to right to identify
$\EL_{\s{x}}$ and $\EL_{\s{y}}$ and compute $\mathcal{C}_{\s{x}}$ and
$\mathcal{C}_{\s{y}}$. At this point, if we can determine the order using
conditions (C1) and/or (C2) of Lemma \ref{lem-a}, then we terminate immediately
returning the order.

\item[Step 2:] We compute the first mismatch position, $h$, between
\s{x} and \s{y}; that is, for $1 \leq i < h$, we have $\s{x_i} = \s{y_i}$ and
$\s{x_h} \neq \s{y_h}$. Now, by applying
Lemma \ref{lem-addLambda},  
we can ignore the letters to its left, because they are equal in $\s{x}$
and $\s{y}$. Note that the case when $h$ lies within $\s{x}_{0} (\s{y}_{0})$ can
be handled easily.

\item[Step 3:] Assume that the nearest $\EL_{\s{x}} = \EL_{\s{y}}$ to the right
of $h$ is at position $\ell_x+1$ ($\ell_y+1$) in $\s{x}$ ($\s{y}$). The case
when $h$ lies within $\s{x}_{j} (\s{y}_{j})$ again can be handled easily.

\item[Step 4:] Now we focus on $\s{x'}=\s{x}_h..\s{x}_{\ell_x} $ and
$\s{y'} =\s{y}_h..\s{y}_{\ell_y}$. Essentially, we will construct a map as
is done in the ADRS algorithm. But we will not construct the map completely;
rather we will construct only the part of the map that is relevant to the
computation in a different way. To do this we count the number of occurrences of each
letter $\alpha \in \Sigma$ within an appropriate range as follows. We start
with the highest letter and continue downward. Assuming that $\sigma =
|\Sigma|$, we use two $\sigma$-length arrays $count_{\s{x}}[1..\sigma]$ and
$count_{\s{y}}[1..\sigma]$ as follows.
Suppose we are counting the number of $\alpha \in \Sigma$. Then we check
the leftmost occurrence $p$ of $\beta > \alpha$ in the range
$\s{x}[h..\ell_x]$ such that there is no occurrence of $\gamma > \beta$
before $p$. And we count the number of occurrences of $\alpha$ in the
range $\s{x}[h..p-1]$ and store it in $count_{\s{x}}[\alpha]$. Similarly
we compute $count_{\s{y}}[\alpha]$.

\item[Step 5:] At this point, in $count_{\s{x}}[1..\sigma]\ (count_{\s{y}}[1..\sigma])$
we have the frequency of each letter $\alpha\in \Sigma$ in the appropriate
range. Now the rest is quite easy. We scan $count_{\s{x}}$, $count_{\s{y}}$ from
the higher to lower letters of $\Sigma$ as follows:
\begin{algorithmic}
\For{$\alpha = highest(\Sigma)~ \mathrm{to}~ lowest(\Sigma)$}
    \If {$count_{\s{x}}[\alpha] == count_{\s{y}}[\alpha]$}
    \LineComment{This means either $\alpha$ is nonexistent (when count is
    zero) or we are in Condition (C3). So we need to check the next letter.}
    \State \bcontinue
	\Else
    \LineComment{If $count_{\s{x}}[\alpha] \not= count_{\s{y}}[\alpha]$,
    then either $\alpha$ is nonexistent in \s{x} --- when
	$count_{\s{x}}[\alpha]$ is zero  --- or in \s{y} --- when
	$count_{\s{y}}[\alpha]$ is zero. That is, we are in Condition (C1) or
	(C2). So we have $count_{\s{x}}[\alpha] <
	count_{\s{y}}[\alpha]$ ($count_{\s{y}}[\alpha] < count_{\s{x}}[\alpha]$,
	respectively).}
	
	\State \Return $\s{x} \prec \s{y}$ ($\s{y} \prec \s{x}$, respectively)
    \EndIf

\EndFor
\end{algorithmic}
\end{description}

At this point a brief discussion is in order. Recall that the ADRS algorithm
runs in $O(n + \sigma)$ time. Because $\sigma$ is $O(n)$, this running time is
optimal. Therefore, we cannot get improvement asymptotically and the theoretical
time complexity of the new algorithm matches that of the ADRS algorithm.
However, the use of Lemma \ref{lem-addLambda} gives us an opportunity to work
much less from a practical point of view, especially for favourable input
strings. And this is why, despite the same theoretical time complexity, our new 
algorithm is an \itbf{input sensitive} algorithm and in practice should perform
better than the ADRS algorithm. 

\subsection{Suffix sorting and Burrows Wheeler transformation}\label{sec:ss-bwt}
 
The suffix permutation \cite{DuvalL02} of a word $\s{w} = w_1 w_2 \ldots w_n$ is
the permutation $\pi_{\s{w}}$ over $\{1, \ldots ,n\}$, where $\pi_{w_i}$ is the
rank of the suffix $\s{w}[i,n]$ in the set of the lexicographically sorted
suffixes of $\s{w}$.
In \cite{HR03} it is shown how to deduce the Lyndon factorization (Theorem
\ref{lLyndon-thrm}) of a text from its suffix permutation; conversely, a
strategy is given in \cite{MRRS14} for obtaining the suffix array from the
Lyndon factorization of a text.

We will outline how our new results from Section \ref{sect-results} can be
applied to obtaining a lex-extension suffix array from the $V$-order
factorization of a text -- the distinctness of factors in a Lyndon versus
$V$-order factorization of a given string \cite{DDS11, DDS13} opens more avenues
for string processing (such as choosing the factorization with more/less factors
for efficiency).

To elaborate, there are three main cases to be handled for the
$V$-factorization algorithm VF in \cite{DDS11, DDS13} as follows. To
determine the $V$-order factorization $\s{x_1} \ge \s{x_2} \cdots \ge
\s{x_k}$ of a string $\s{x}$, algorithm VF applies Lemma 3.16 in
\cite{DDS13} to substrings \s{x_i}, \s{x_j}:

\begin{itemize}
\item If (C1) holds for \s{x_i}, \s{x_j} ($\mathcal{L_{\s{x_i}}} <
\mathcal{L_{\s{x_j}}}$) then $\s{x_i} > \s{x_j}$ in the factorization - the
algorithm tracks maximal elements.

\item  If (C2) holds for \s{x_i}, \s{x_j} then, $\s{x_i} < \s{x_j}$ if $\s{x_i}
\s{x_j}$ is a Hybrid Lyndon (that is a Lyndon word under lex-extension
\cite{DDS13}), and $\s{x_i} \s{x_j}$ is a factor in the factorization --  the
algorithm checks for concatenating repetitions.

\item  If (C3) holds for \s{x_i}, \s{x_j}, and if $\s{x_i} \prec \s{x_j}$ then
$\s{x_i} \s{x_j}$ is a factor in the factorization --  the algorithm compares
substrings between maximal elements.

\end{itemize}

As each factor is identified by algorithm VF, its rightmost position is recorded
(procedure output) and then all housekeeping variables are re-initialized
(procedure RESET) -- this essentially converts the remaining suffix of the
string into a new string to be factored with no re-visiting of the previously
factored elements required. Hence, similarly to Duval's Lyndon decomposition
algorithm \cite{Duval83}, the linear $V$-order factoring technique can be used
for on-line scenarios which is the setting of our applications.

Now, we are interested in the notion of compatibility for sorting suffixes as
introduced in \cite{MRRS14}. Let $\s{x}$ be a word and \s{u} be a substring
(factor) of $\s{x}$. The sorting of suffixes $\s{s_1}$, $\s{s_2}$ of \s{u}, with
respect to \s{u}, is {\it compatible} with the sorting of the suffixes of
$\s{x}$ for which $\s{s_1}$, $\s{s_2}$ are prefixes, with respect to $\s{x}$, if
they have the same order in both \s{u} and \s{x}.
It is shown in \cite{MRRS14} that, although compatibility doesn't always hold
for lexorder suffix-sorting, when \s{u} is chosen to be a substring of Lyndon
factors in a factorization then it does hold. In contrast, compatibility always
holds for sorting suffixes in $V$-order, and furthermore, the shorter suffix is
always lesser:

\begin{lemma}
\label{lem-compatible}
Let $\s{x} \in \Sigma^+$  and \s{u} be a substring of $\s{x}$ with $\s{s_1}$ a
suffix of $\s{u}$. If $\s{s_2}$ is a suffix of $\s{s_1}$ then $\s{s_2} \prec
\s{s_1}$ with respect to both \s{u} and \s{x}.
\end{lemma}

\begin{proof}
Consider the suffixes $\s{s_1 t_1}$ and $\s{s_2 t_2}$ of $\s{x}$ for possibly
empty $\s{t_1}$, $\s{t_2}$.
Applying Lemma \ref{lem-subsequce} then both $\s{s_2} \prec \s{s_1}$ with
respect to \s{u} and $\s{s_2 t_2} \prec \s{s_1 t_1}$ with respect to \s{x}. 
\qed
\end{proof}

Lemma \ref{lem-subsequce} further shows that suffixes are totally $V$-ordered by
their given order: for any string $\s{x} = \s{x}[1 \ldots n]$, we have $x_n
\prec x_{n-1}x_n \prec \cdots \prec \s{x}$.

However, to address applications involving conjugates of strings, such as the
Burrows Wheeler transform, Lemma \ref{lem-compatible} doesn't suffice for
$V$-order: when using suffixes to sort all rotations of a string, since each
rotation has the same number of maximal elements, therefore implicitly condition
(C3) applies --- for ordering these suffixes we need the first distinct prefix
substrings of the $V$-forms of the suffixes. We will use {\it lex-extension
ordering} which compares factors in a factorization pair-wise from left to right
while each comparison is made in $V$-order.

\begin{theorem}
\label{cor-V_factorization+suffixes}
Let $\s{x} \in \Sigma^+$ with $V$-order factorization $\s{x} = \s{x_1} \cdots
\s{x_k}$, and let $\s{u} = \s{x_i} \cdots \s{x_j}$, for $1 \le i \le j \le k$.
Then the sorting of the suffixes of \s{u} is compatible with the sorting of the
suffixes of \s{x}.
\end{theorem}

\begin{proof}
The case of the Lyndon factorization is Theorem 3.2 in \cite{MRRS14}. The
$V$-order proof thus follows from the Lyndon-like properties of the $V$-order
factorization and by replacing lexorder with lex-extension ordering.
\qed
\end{proof}

Equipped with this theorem, the clever incremental suffix sorting \& BWT
strategy introduced in \cite{MRRS14} can be modified for $V$-order: 

\begin{description}\setlength\itemsep{0.3em} 
\item[Step 1:] Compute the $V$-order factorization of $\s{x} = \s{v_1} \cdots
\s{v_k}$ in linear time \cite{DDS11, DDS13}.
\item[Step 2:] Compute the lex-extension order suffix array of each of $\s{v_1}$
and $\s{v_2}$ in linear time \cite{DS14}.
\item[Step 3:] Obtain the BWT($\s{v_i}$) from each SA($\s{v_i}$): for a suffix
$\s{v_i}=\s{x}[h \ldots m]$ the BWT character is $\s{x}[h-1]$.
\item[Step 4:]  Merge the sorted suffixes in Step 2 using ADRS algorithm
\cite{ADRS14FI} to obtain the suffix array of $\s{v_1} \s{v_2}$. For the merge,
if $\s{v_j} \succ \s{v_k}$, then the chosen suffix for the new array is
$\s{v_k}$, otherwise it is $\s{v_j v_k}$.
\item[Step 5:] Obtain the BWT of the merged sorted suffixes in Step 4. If the
chosen suffix for the new array was $\s{v_k}$, then the BWT character is given
by BWT($\s{v_k}$); otherwise it is BWT($\s{v_j}$) since the prefix $\s{x}[1
\ldots h-1]$ in \s{x} is rotated as $\s{v_j v_k} \dots \s{x}[1 \ldots h-1]$.
\item[Step 6:]  Compute the lex-extension order suffix array of $\s{v_3}$ and
merge it with the suffix array of $\s{v_1} \s{v_2}$ from Step 4 and obtain the
BWT.
\item[Step 7:] Repeat until all the $V$-factors have been incrementally processed.
\end{description} 

Overall, for iterating over $k$ factors, the time complexity is $O(k^2 n)$, with
each iteration taking $O(k n)$. As expressed in \cite{MRRS14} for the Lyndon
case, this technique is suitable for integration with the on-line $V$-order
factoring algorithm: suffix sorting can proceed in tandem as soon as the first
$V$-factor is identified. Note that in Step 4 above, the new string comparison
algorithm presented in Section \ref{sec:newAlgo} can be applied when
input-sensitivity is relevant.

\section{Future Research}\label{sect-future}

We propose the following problem: Suppose that $\s{x}, \s{y} \in \Sigma^+$ with
$\s{x} \prec \s{y}$. Under what permutations $\pi$, that is, $\s{x} \rightarrow
\pi(\s{x})$ and $\s{y} \rightarrow \pi(\s{y})$  does $\pi(\s{x}) \prec
\pi(\s{y})$ hold? For instance, for integers, $21 \prec 12$ and no permutation
works; whereas interchanging the first and last letters does for $142 \prec
243$ since  $241 \prec 342$, which generalizes to requiring that the rightmost
substrings of their $V$-forms are in $V$-order.

We propose studying such permutations in the context of the gene team problem: 
to find a set of genes that appear in two or more species, possibly in a
different order, but within a given distance in each chromosome -- this has
impact in understanding genome evolution and function
\cite{10.1109/TCBB.2010.127}.  
  
\bibliographystyle{plain}
\bibliography{references}
\end{document}